\documentclass[twocolumn]{autart}

\usepackage{graphicx}
\usepackage{amsmath,amssymb,amsfonts,mathrsfs}
\usepackage{xcolor,epstopdf}
\usepackage{enumitem}
\usepackage{bbold}
\usepackage{algorithm}
\usepackage{algpseudocode}
\usepackage{url}
\usepackage{comment}

\newtheorem{proof}{Proof}
\newtheorem{definition}{Definition}
\newtheorem{remark}{Remark}
\newtheorem{lemma}{Lemma}

\newtheorem{theorem}{Theorem}

\renewenvironment{proof}{\textit{Proof.}}{\hfill $\square$}

\newcommand{\smat}[1]{\ensuremath{\left[\begin{smallmatrix}#1\end{smallmatrix}\right]}}
\newcommand{\bmat}[1]{\ensuremath{\begin{bmatrix}#1\end{bmatrix}}}
\newcommand{\tu}[1]{\textup{#1}}
\newcommand{\sa}[1]{\mathsf{#1}}

\newcommand{\hJac}{\ensuremath{\frac{\partial h}{\partial x}(x)}}
\newcommand{\hJacS}{\ensuremath{\tfrac{\partial h}{\partial x}(x)}}

\newcommand{\real}{\ensuremath{\mathbb{R}}}

\DeclareMathSymbol{\cdoT}{\mathord}{symbols}{"01}

\allowdisplaybreaks

\begin{document}

\begin{frontmatter}
\title{Data-driven design of safe control for polynomial systems} 

\thanks{This research is conducted under the auspices of the Centre for Data Science and Systems Complexity at the University of Groningen and is supported by a Marie Sk\l{}odowska-Curie COFUND grant, no. 754315.}

\author[rug]{Alessandro Luppi}\ead{a.luppi@rug.nl},    
\author[rug]{Andrea Bisoffi}\ead{a.bisoffi@rug.nl},              
\author[rug]{Claudio De Persis}\ead{c.de.persis@rug.nl},
\author[flo]{Pietro Tesi}\ead{pietro.tesi@unifi.it}  
  
\address[rug]{ENTEG and J.C. Willems Center for Systems and Control, University of Groningen, 9747 AG Groningen, The Netherlands}
\address[flo]{DINFO, University of Florence, 50139 Florence, Italy}


\begin{abstract}                          
We consider the problem of designing an invariant set using only a finite set of input-state data collected from an unknown polynomial system in continuous time.
We consider noisy data, i.e., corrupted by an unknown-but-bounded disturbance.
We derive a data-dependent sum-of-squares program that enforces invariance of a set and also optimizes the size of the invariant set while keeping it within a set of user-defined safety constraints; the solution of this program directly provides a polynomial invariant set and a state-feedback controller.
We numerically test the design on a system of two platooning cars. 
\end{abstract}
\end{frontmatter}

\section{Introduction}
\label{sec:intro}	
Control systems are nowadays required not only to stabilize the process around a desired working point given by an equilibrium, but also to be able to steer away the state of the system from dangerous regions where it would not work properly.
The complement of this dangerous set is the so-called safe set and the design of control schemes to respect safety specifications goes by the name of safe control \cite{garone2017reference,wolff2005invariance,wabersich2018scalable,cortez2021safebydesign}.
The notion of safety is intuitively related to the notion of invariance, namely, the dynamical propriety that the state belongs to a certain set for all subsequent times after being initialized in there. 
Historically, a seminal result for invariance was Nagumo's theorem \cite{nagumo1942lage}, which has been of fundamental importance in the characterization of invariant sets for continuous-time systems. 
Since almost every system in practice is subject to some type of constraints on its states or outputs, the notion of invariance is very relevant in control applications to include safety constraints in the design. 
Indeed, problems related to safety and viability \cite{aubin2011viability} can be addressed by computing sets possessing invariance or closely-related properties.

In the control community, invariance for linear systems was extensively studied in the 1980-1990's, which resulted in a number of works surveyed in \cite{blanchini1999set}. For nonlinear systems, there has recently been a revived interest with the introduction of control Lyapunov-like functions tailored to enforce invariance, called control barrier functions \cite{xu2015robustness,xu2017correctness,wieland2007constructive,prajna2004safety}.
In this setting, the controller depends critically on the model used for design and recent studies \cite{jankovic2018robust,garg2021robust,lopez2020robust} focus on finding robust safe controllers to account for possibly inaccurate models. 
	In \cite{alan2021safe}, unmodeled dynamics are taken into account by adding a bounded disturbance on the input to find a robust control barrier function, in an input-to-state-safety fashion.

In an effort to reduce this critical dependence from the model, we implement a data-driven solution for safe control.
This work follows on from~\cite{bisoffi2020controller} on enforcing invariance with a controller designed directly from data and it extends \cite{bisoffi2020controller} to the case of polynomial nonlinear systems and by dropping the assumption that the set to make invariant is given.
Polynomial systems are a notable class of nonlinear systems widely used to model processes in engineering applications such as fluid dynamics \cite{chernyshenko2014polynomial,lakshmi2020finding} and robotics \cite{majumdar2013control}.
Casting nonlinear control design as an optimization problem faces the obstacle that it is generally computationally intractable to verify whether a multivariate (matrix) polynomial is nonnegative.
For a computationally viable approach, one can adopt a relaxation and verify instead whether such a polynomial is a sum-of-squares (SOS) since SOS optimization can be solved through semidefinite programming (SDP) \cite{parrilo2000structured,chesi2010lmi}.
In a model-based setting, SOS programs have been used to design stabilizing controllers for hybrid systems \cite{papachristodoulou2005tutorial}, for disturbance analysis in linear systems \cite{jarvis2005control} and to obtain inner-approximations of the basin of attraction \cite{tan2008stability}, to name a few applications.

Still, SOS programs suffer from some limitations that have been or are being addressed in the literature.
A first one is scalability of SOS programs: tools are available \cite{Lofberg2009} that automatically reduce the problem size without major computational costs, and recent works on large-scale semidefinite and polynomial optimization \cite{zheng2021chordal} improve scalability of SOS programs significantly.
Secondly, it is often the case that the obtained SOS program is bilinear in the decision variables. This occurs in the model-based case \cite{tan2008stability} and also in this paper. An iterative approach to solve these bilinear SOS programs is commonly used \cite{jarvis2005control}. Alternatively, there exist tools to solve these bilinear SOS programs directly, such as PENBMI and BMIBNB.
For these reasons, the limitations of SOS programs seem largely outweighed by their positive features.

From the side of (direct) data-driven control, this work is positioned in the literature thread \cite{markovsky2008data,coulson2019data,depersis2019formulas}, to name a few, which exploits the so-called fundamental lemma in~\cite[Th.~1]{willems2005note}.
Within this thread, \cite{bisoffi2020controller,dai2021semi,Guo2020,bisoffi2021petersen} are the most closely related works to this one.
As mentioned before, \cite{bisoffi2020controller} addresses also an invariance problem, but for systems with \emph{linear} dynamics and where the set to be rendered invariant is \emph{given}. The works \cite{dai2021semi,Guo2020} and \cite[\S 5]{bisoffi2021petersen} consider nonlinear input-affine polynomial systems as here, but the goal is data-driven almost global \cite{dai2021semi} or global \cite{Guo2020,bisoffi2021petersen} stabilization;
invariance is a fairly weaker dynamical property (e.g., solutions do not need to converge to an attractor within the invariant set), hence the conditions here are fairly less conservative and yet significant to enforce safety.

\textit{Contribution.} Our contribution is to enforce invariance using no model of the system to be controlled, but only a set of input-state data points collected from it in an experiment. 
We consider an input-affine nonlinear system with polynomial dynamics and a polynomial controller. 
This allows us to make the data-based invariance conditions tractable by using an SOS  relaxation and alternately solving two SOS programs.
In this work we consider the realistic setting when invariance needs to be guaranteed despite the presence of an unknown-but-bounded \cite{hjalmarsson1993discussion} additive noise in data.
Moreover, we show that also in the data-driven case it is possible to optimize the size of the invariant set while respecting safety constraints.
Finally, we provide numerical evidence to show the effectiveness of the approach, in particular on the practical example of two platooning cars.

\textit{Structure.} 
Section~\ref{sec:preli} recalls the theory enabling our results. 
In Section~\ref{sec:invProblem}, we set up the invariance problem for polynomial systems. 
In Section~\ref{sec:data_driven}, we derive a data-driven controller that solves the invariance problem, and also includes safety constraints.
In Section~\ref{sec:platoon}, we exemplify our data-based solution for two platooning cars.

\textit{Notation.}	For a matrix $A$, $A^{\top}$ denotes its transpose. 
A symmetric matrix $A \in \mathbb{R}^{n \times n}$ is positive semidefinite, i.e., $A \succeq 0$, if $x^{\top}Ax \geq 0$ for all $x \in \mathbb{R}^n$. 
Given two matrices $A$ and $B$, $A \succeq B$ means $A-B \succeq 0$. 
With $\mathbb{Z}_{\ge 0}$ we indicate the set of nonnegative integers. 
An identity matrix is denoted by $I$. 
For a positive semidefinite $A=A^\top$, $A^{1/2}$ denotes the unique positive semidefinite root of $A$.
For matrices $M$ and $N=N^\top$, we sometimes abbreviate $M N M^\top$ as $M N M^\top = M \cdot N [\star]^{\top}$.
For symmetric matrices $A$ and $C$, we sometimes abbreviate the symmetric matrix $\smat{A & B^\top\\ B & C}$ as $\smat{A & \star \\ B & C}$ or $\smat{A & B^\top\\ \star & C}$.

\section{Preliminaries}\label{sec:preli}	

In this section, we present a fundamental result from real algebraic geometry, the Positivstellensatz. We first report from \cite{chesi2010lmi,jarvis2005control} the notions needed to state that result.

We start from sum-of-squares (SOS) polynomials and SOS matrix polynomials.
A function $h: \mathbb{R}^n \rightarrow \mathbb{R}$ is a \textit{monomial} of degree $d$ in $x = (x_1, x_2, \dots, x_n) \in \mathbb{R}^n$ if 
\begin{equation*}
	h(x) = a x_1^{q_1} x_2^{q_2} \cdots x_n^{q_n}
\end{equation*}
with $a \in \mathbb{R}$, $q_1, \dots , q_n \in \mathbb{Z}_{\ge 0}$ and $d = \sum_{i=1}^{n} q_i$. 
A function $h: \mathbb{R}^n \rightarrow \mathbb{R}$ is a \textit{polynomial} if it is a sum of (a finite number of) monomials $h_1$, $h_2, \ldots : \mathbb{R}^n \rightarrow \mathbb{R}$ with finite degree, and the largest degree of the $h_i$'s is the degree of $h$. $\Pi$ denotes the set of polynomials.

\begin{definition}[SOS polynomial]
	$h \in \Pi$ is an \emph{SOS} polynomial if there exist $h_1, \dots, h_k \in \Pi$ such that $h(x) = \sum_{i=1}^{k} h_i(x)^2$. The set of SOS polynomials $h \in \Pi$ is denoted as $\Sigma$.
\end{definition}

A function $H : \real^n \to \real^{r_1 \times r_2}$ is a \emph{matrix polynomial} if the entries of $H$ satisfy $h_{ij} \in \Pi$ for all $i = 1$, \dots, $r_1$ and $j = 1$, \dots, $r_2$, and the largest degree of the entries of $H$ is the degree of $H$.
The set of matrix polynomials $H : \real^n \to \real^{r_1 \times r_2}$ is denoted by $\Pi_{r_1,r_2}$.
A function $H : \real^n \to \real^{r_1 \times r_2}$ is a \emph{square matrix polynomial} if $r_1 = r_2$.
The set of square matrix polynomials $H : \real^n \to \real^{r \times r}$ is denoted by $\Pi_r$.

\begin{definition}[SOS matrix polynomial \cite{chesi2010lmi}]\label{def:matrixSOS}
	$H \in \Pi_r$ is an \emph{SOS matrix polynomial} if there exist $H_1$, \dots, $H_k \in \Pi_r$ such that $H(x) = \sum_{i=1}^{k} H_i(x)^{\top} H_i(x)$. The set of SOS matrix polynomials $H \in \Pi_r$ is denoted by $\Sigma_r$.		
\end{definition}

SOS polynomials are instrumental to define three sets of polynomials appearing in the Positivstellensatz.

\begin{definition}[Multiplicative monoid {\cite[Def.~3]{jarvis2005control}}]
	\label{def:mult monoid}
	Given $g_{1}$, \dots, $g_{t} \in \Pi$, the \emph{multiplicative monoid} generated by $g_{j}$'s is the set of all finite products of $g_{j}$'s, including $1$ (i.e., the empty product). It is denoted by $\mathcal{M}(g_1, \dots, g_t)$, with $\mathcal{M}(\emptyset):=1$ for completeness.
\end{definition}
An example is $\mathcal{M}(g_1, g_2)=\{g_1^{k_1} g_2^{k_2} \colon k_1, k_2 \in \mathbb{Z}_{\ge 0}\}$.

\begin{definition}[Cone {\cite[Def.~4]{jarvis2005control}}]
	Given $f_{1}, \dots, f_{r} \in \Pi$, the \emph{cone} generated by $f_{i}$'s is
	$\mathcal{C} (f_{1}, \ldots, f_{r}):=  \big\{s_{0}+\sum_{i=1}^{l} s_{i} b_{i} \colon l \in \mathbb{Z}_{\ge 0}, s_{i} \in \Sigma, b_{i} \in \mathcal{M}(f_{1}, \ldots, f_{r})\big\}$.
\end{definition}

If $s \in \Sigma$ and $f \in \Pi$, then $f^{2} s \in \Sigma$. Then, a cone of $\{f_{1}, \dots, f_{r}\}$ can be expressed as a sum of $2^{r}$ terms without loss of generality.
An example is $\mathcal{C}(f_{1}, f_{2})=$ $\{s_{0}+s_{1} f_{1}+s_{2} f_{2}+s_{3} f_{1} f_{2} \colon s_{0}, \dots, s_{3} \in \Sigma\}$ where terms like $s_4 f_1^2$ or $s_5 f_2^2$ with $s_4$, $s_5 \in \Sigma$ are not needed since they are captured by $s_0$ anyway.

\begin{definition}[Ideal {\cite[Def.~5]{jarvis2005control}}]
	\label{def:ideal}
	Given $h_{1}, \dots, h_{u} \in \Pi$, the \emph{ideal} generated by $h_k$'s is
	$\mathcal{J}(h_{1}, \dots, h_{u}):=\{\sum_{k=1}^{u} h_k p_k \colon p_k \in \Pi\}$.
\end{definition}
An example is $\mathcal{J}(h_1,h_2)=\{ h_1 p_1 + h_2 p_2 \colon p_1, p_2 \in \Pi\}$. 
With Definitions~\ref{def:mult monoid}-\ref{def:ideal}, we finally recall the version in~\cite{jarvis2005control} of the Positivstellensatz (P-Satz), in the next fact.

\begin{fact}[P-Satz {\cite[Th.~1]{jarvis2005control}}]
	\label{fact:psatz}
	Given $f_{1}, \dots, f_{r} \in \Pi$, $g_{1}, \dots, g_{t} \in \Pi$, and $h_{1}, \dots, h_{u} \in \Pi$, the following are equivalent.
	\begin{enumerate}
		\item The set
		\begin{equation*}
		\left\{x \in \mathbb{R}^{n} \colon \begin{array}{c}
			f_{1}(x) \geq 0, \dots, f_{r}(x) \geq 0 \\
			g_{1}(x) \neq 0, \dots, g_{t}(x) \neq 0 \\
			h_{1}(x)=0, \dots, h_{u}(x)=0
		\end{array}\right\} = \emptyset.
		\end{equation*}
		\item 	There exist polynomials $f \in \mathcal{C}(f_{1}, \dots, f_{r})$, $g \in \mathcal{M}(g_{1}, \dots, g_{t})$, $h \in \mathcal{J}(h_{1}, \dots, h_{u})$ such that
		\begin{equation*}
		f+g^{2}+h=0.
		\end{equation*}
	\end{enumerate}
\end{fact}

\section{Collecting data and enforcing invariance}
\label{sec:invProblem}

In this paper we consider polynomial systems of the form
\begin{equation}\label{eq:polysys}
	\dot{x} = f(x) + g(x)u
\end{equation}
where $x \in \mathbb{R}^{n}$ is the state, $u \in \mathbb{R}^m$ is the control input, $f$ and $g$ are polynomial vector fields. 
The specific expressions of $f$ and $g$ are unknown. 
The polynomial system~\eqref{eq:polysys} can be written into the linear-like form
\begin{equation}\label{polynSysOpen}
	\dot{x} = \sa{A}Z(x) + \sa{B}W(x)u
\end{equation}
where $\sa{A}\in \real^{n \times N_\sa{A}}$ and $\sa{B} \in \real^{n \times N_\sa{B}}$ are \emph{unknown} constant matrices, the \emph{known} $N_\sa{A} \times 1$ vector $Z(x)$ contains as entries the distinct monomials in $x$ that may appear in $f$, and the \emph{known} $N_\sa{B} \times m$ matrix $W(x)$ contains as entries the monomials that may appear in $g$. 
The conditions we will propose to design an invariant set are the same regardless of the choice of the monomials in $Z$ and $W$.
On the other hand, different choices of $Z$ and $W$ affect feasibility and goodness of the solution arising from those conditions, as is generally the case with model structure selection.
In Section~\ref{sec:platoon}, we will present guidelines for the choice of monomials in $Z$ and $W$.

We consider the control law $u = K(x)$ where $K \in \Pi_{m, 1}$ is to be designed.
The closed-loop dynamics results in
\begin{align}
	\!\dot{x} & = \sa{A} Z(x) + \sa{B} W(x) K(x) = \bmat{\sa{A} & \sa{B}} \bmat{Z(x) \\ W(x) K(x)} . \label{polynSysClosed}
\end{align}

Data are generated through an experiment in the presence of an additive disturbance $d$ as
\begin{equation}
	\label{polynSysDataGen}
	\dot x=  \sa{A} Z(x) + \sa{B} W(x) u + d.
\end{equation}
We apply an input sequence of $T$ elements, and measure the state and state derivative sequences generated by~\eqref{polynSysDataGen}; we sample uniformly these sequences at times $0$, $\tau_{\tu{s}}$, \dots, $(T-1)\tau_{\tu{s}}$ for sampling period $\tau_{\tu{s}} > 0$; this results in data points, for $j = 0, \dots, T-1$,
\begin{equation}
	\label{data instant}
	\dot{x}^j:=\dot{x}(j\tau_{\tu{s}}), \, z^j:= Z(x(j\tau_{\tu{s}})),\, v^j:= W(x(j\tau_{\tu{s}}))u(j\tau_{\tu{s}}).
\end{equation}	
A disturbance sequence given for $j = 0, \dots, T-1$ by
\[
d^j:=d(j\tau_{\tu{s}})
\]
acts during the experiment but is \emph{unknown}. Hence, the data generation mechanism is described by
\begin{equation}
	\label{data gen mech}
	\dot{x}^j = \sa{A} z^j + \sa{B} v^j + d^j, j = 0, 1, \dots, T-1.
\end{equation}

Our goal is to use the collected data points to obtain an invariant set for~\eqref{polynSysClosed} as specified in the next definition.

\begin{definition}[Invariant set]\label{def:invariant}
	For $a \colon \real^n \to \real^n$ polynomial, i.e., $a \in \Pi_{n,1}$, and for an arbitrary $x_0 \in \real^{n}$, denote $t \mapsto \alpha(t,x_0)$ the (unique maximal) solution to $\dot{x}=a(x)$ with initial condition $x_0= \alpha(0,x_0)$ and defined on the interval $[0,T(x_0))$ (with $T(x_0)$ possibly $+\infty$). A set $\mathcal{I}$ is said to be \emph{invariant} for $\dot{x}=a(x)$ if $x_0 \in \mathcal{I}$ implies that for all $t \in [0, T(x_0))$, $\alpha(t,x_0) \in \mathcal{I}$.		
\end{definition}

Consider the set
\begin{equation}\label{eq:invset}
	\mathcal{I}:=\{ x \in \mathbb{R}^n : h(x) \leq 0 \}
\end{equation}
with $h \in \Pi$. To impose that $\mathcal{I}$ is an invariant set according to Definition~\ref{def:invariant}, we require the condition
\begin{equation}\label{eq:inclusion}
	\{ x \in \mathbb{R}^n : h(x) = 0\} \subseteq \{ x \in \mathbb{R}^n : \frac{\partial h}{\partial x}(x)\dot{x} \leq -\epsilon\}
\end{equation}
where $\dot{x}$, used for brevity, takes the expression in~\eqref{polynSysClosed} and the parameter $\epsilon > 0$ is introduced to guarantee some degree of robustness at the boundary of $\mathcal{I}$.

\begin{remark}
	\label{remark:degree of rob}
	We emphasize that whereas we use noisy data for control design as per the data generation mechanism in~\eqref{data gen mech}, we design a controller to enforce \emph{nominal} invariance for $d=0$ as in~\eqref{polynSysClosed}, instead of robust invariance.
	Nonetheless, our design features some degree of robustness thanks to $\epsilon$, as we now show.
	Define the nominal closed-loop vector field in~\eqref{polynSysClosed} as $f_{\tu{n}}(x) :=\sa{A} Z(x) + \sa{B} W(x) K(x)$ and consider now the perturbed dynamics $\dot x = f_{\tu{n}}(x) + d$ with $|d|^2 \le \omega$ (as we will have later in~\eqref{in:bound}).
	$\mathcal{I}$ is robustly invariant for this perturbed dynamics as long as 
	\begin{equation}
		\label{cond for rob inv}
		\hJacS ( f_{\tu{n}} (x) + d ) \le 0 \quad \forall (x,d) \colon h(x) = 0, |d|^2 \le \omega
	\end{equation}
	By achieving \eqref{eq:inclusion}, we have that for all $x$ and $d$ with $h(x) = 0$ and $|d|^2 \le \omega$,
	\begin{equation*}
		\hJacS ( f_{\tu{n}} (x) + d ) \le -\epsilon + \left| \hJacS \right| \sqrt{\omega}.
	\end{equation*}
	Then, robust invariance in~\eqref{cond for rob inv} is achieved if $\left| \hJacS \right| \le \epsilon/\sqrt{\omega}$ for all $x$ such that $h(x) = 0$ or, equivalently, if
	\begin{equation*}
		\bmat{-\epsilon^2/\omega & \hJacS\\ \hJacS^\top & -I} \preceq 0 \quad \forall x \colon h(x) = 0,
	\end{equation*}
	which can be relaxed into an SOS condition and added to our optimization program (see later Theorem~\ref{cor:extra conditions}).
	For these reasons, we rather consider for simplicity nominal invariance and introduce $\epsilon > 0$ to guarantee nonetheless some degree of robustness at the boundary of $\mathcal{I}$.
\end{remark}

With the goal of applying Fact~\ref{fact:psatz}, \eqref{eq:inclusion} can be cast as
\begin{equation*}
	\begin{split}
		\Big\{ x \in \mathbb{R}^n :\ & h(x) = 0, \\
		& \frac{\partial h}{\partial x}(x)\dot{x} + \epsilon \geq 0,\ \frac{\partial h}{\partial x}(x)\dot{x} + \epsilon \neq 0 \Big\} = \emptyset.
	\end{split}
\end{equation*}
This empty set condition is equivalent, by Fact~\ref{fact:psatz}, to the existence of $\ell_1 \in \Pi$, $s_0,s_1 \in \Sigma$ and $k_1\in \mathbb{Z}_{\ge 0}$ such that
\begin{equation}\label{eq:invsos}
	\ell_1 h + s_0 + s_1\Big(\frac{\partial h}{\partial x}\dot{x} + \epsilon\Big) + \Big(\frac{\partial h}{\partial x}\dot{x} + \epsilon\Big)^{2k_1} = 0.
\end{equation}
With the final goal of implementing numerically this condition, we simplify \eqref{eq:invsos} by setting $s_0 = 0$, $k_1=1$ and $\ell_1 = \ell \big( \frac{\partial h}{\partial x}\dot{x} + \epsilon\big)$ and, for $\ell \in \Pi$ and $s_1 \in \Sigma$, obtain
\begin{equation*}
	\Big(\frac{\partial h}{\partial x}\dot{x} + \epsilon\Big) \Big[\ell h + s_1 + \Big(\frac{\partial h}{\partial x}\dot{x} + \epsilon\Big)  \Big] = 0,
\end{equation*}
which is a relaxation of \eqref{eq:invsos}.
So, the original \eqref{eq:inclusion} is implied by
\begin{equation}\label{invCond}
	\ell(x) h(x) + \frac{\partial h}{\partial x}(x)\dot{x} + \epsilon = -s_1(x) \leq 0 \qquad \forall x.
\end{equation}
The arguments above can be summarized as follows.
\begin{lemma}
	If there exists $\ell \in \Pi$ such that condition \eqref{invCond} holds, then $\mathcal{I}$ in \eqref{eq:invset} is an invariant set for \eqref{polynSysDataGen}.
\end{lemma}

Indeed, for all $x$ such that $h(x) =0$ (corresponding to the boundary of $\mathcal{I}$), the condition imposes $\frac{\partial h}{\partial x}(x)\dot{x} \le -\epsilon$, i.e., the Lie derivative of $h$ is strictly negative. 
For all $x$ such that $h(x) <0$ (corresponding to the interior of $\mathcal{I}$), the condition imposes $\frac{\partial h}{\partial x}(x)\dot{x} \le -\epsilon  - \ell(x)h(x)$. 
Since $\ell$ is a polynomial without any sign definiteness requirement and $\epsilon$ is a design parameter selected as a small positive number, the term $-\epsilon - \ell(x) h(x)$ does not need to be negative and can actually be positive; hence, the condition may allow even for a positive $\frac{\partial h}{\partial x}(x)\dot{x}$, consistently with set invariance being less restrictive than attractivity.

\section{Data-driven safe control}
\label{sec:data_driven}
Our goal is to formulate a condition depending exclusively on noisy data to find an invariant set for the actual system \eqref{polynSysClosed}. We substitute in~\eqref{invCond} the closed-loop dynamics in~\eqref{polynSysClosed} and obtain
\begin{align*}
	\ell(x) h(x)  + \epsilon  + \hJac \bmat{\sa{A} & \sa{B}} \bmat{Z(x)\\ W(x) K(x)} \le 0  \quad \forall x. 
\end{align*}
Since the model is not available and the true coefficient matrices $\sa{A}$ and $\sa{B}$ are unknown, we rather enforce the previous inequality on all matrices $(A,B)$ that are consistent with data (for a given disturbance model), as we now characterize.
By \textit{consistent with data}, we mean all matrices $(A,B)$ that could have produced the measured data sequences $\{ \dot{x}^j, z^j, v^j\}_{j=0}^{T-1}$ as in~\eqref{data instant} for an additive disturbance $d$ that is bounded.
A realistic bound on disturbance is that the norm of any possible disturbance instance $d$ is upper-bounded, and so are the norms of the unknown $d^0$, \dots, $d^{T-1}$. This corresponds to an \emph{instantaneous} bound given, for $\omega \geq 0$, by the set
\begin{equation}
	\mathcal{D}_{\tu{i}} := \{d \in \mathbb{R}^n : |d|^2 \leq \omega\}. \label{in:bound}
\end{equation}
Then, based on the bound in $\mathcal{D}_{\tu{i}}$, the matrices $(A,B)$ consistent with a single data point $j \in \{0,\dots,T-1\}$ belong to the set
\begin{align}
	\mathcal{C}^j_{\tu{i}} := &  \Big\{
	[A\ B] \colon   \dot{x}^j = A z^j + B v^j + d,\,
	d d^\top \preceq \omega I
	\Big\}, \label{Ci set}\\
	= &  \Big\{ [A\ B] \colon 
	\bmat{I & A & B}
	\bmat{I & \dot{x}^j\\ 0 & -z^j \\ 0 & -v^j}  \cdoT 
	\bmat{\omega I & 0\\ 0 & -I}
	[\star]^\top \! \succeq 0
	\Big\}. \notag
\end{align}
$\mathcal{C}^j_{\tu{i}}$ is the set of all matrices for which some disturbance $d$ satisfying the bound in~\eqref{in:bound} could have generated the measured data point $j$, as in~\cite{milanese2004set,dai2018moments}. 
The set of matrices $(A,B)$ consistent with all data points $j=0, \dots, T-1$ is then
\begin{equation}
	\label{inter Ci set}
	\mathcal{C}_{\tu{i}}:=\bigcap_{j=0}^{T-1} \mathcal{C}^j_{\tu{i}}.
\end{equation}%
We discussed in \cite{bisoffi2021tradeoffs} that, for matrices $\zeta_0$, $P=P^\top \succ 0$, $Q=Q^\top \succeq 0$, the set $\mathcal{C}_{\tu{i}}$ cannot be expressed as a matrix ellipsoid\footnote{As the name suggests, a matrix ellipsoid is an extension of the classical (vector) ellipsoid $\{\zeta \in \mathbb{R}^p \colon (\zeta-\zeta_0)^\top P^{-2} (\zeta-\zeta_0 ) \le q \}$ with $P = P^\top \succ 0$ and $q \ge 0$, see \cite{bisoffi2021tradeoffs} for details.} of the form	
\begin{equation}
	\{\zeta \in \mathbb{R}^{(n+m) \times n} \colon (\zeta-\zeta_0)^\top P^{-2} (\zeta-\zeta_0 ) \preceq Q \},
\end{equation}	
which is instrumental to obtain our main result.
On the other hand, this form can be obtained for the matrix ellipsoid $\overline{\mathcal{C}}_{\tu{i}}$ defined for $A_{\tu{i}} \succ 0$ as 
\begin{align}
	& \overline{\mathcal{C}}_{\tu{i}} := \{ [A\ B]= \zeta^\top \colon  \notag \\
	& \hspace*{10mm}		\bmat{I & \zeta^\top}
	\bmat{B_{\tu{i}}^\top A_{\tu{i}}^{-1} B_{\tu{i}} - I~ & B_{\tu{i}}^\top\\
		B_{\tu{i}}~ & A_{\tu{i}}}
	\bmat{I\\ \zeta} \preceq 0 \} .\label{setOverlineI}
\end{align}
Indeed, for $A_{\tu{i}} \succ 0$, the condition in~\eqref{setOverlineI} is precisely $(\zeta + A_{\tu{i}}^{-1}B_{\tu{i}})^\top A_{\tu{i}} (\zeta + A_{\tu{i}}^{-1}B_{\tu{i}} ) \preceq I$, so that the selections
\begin{equation}
	\label{ell matr equiv}
	\overline{\zeta}_{\tu{i}} := - A_{\tu{i}}^{-1}B_{\tu{i}},\quad \overline{P}_{\tu{i}}^{-2} := A_{\tu{i}},\quad \overline{Q}_{\tu{i}}:=I
\end{equation}
rewrite $\overline{\mathcal{C}}_{\tu{i}}$ equivalently as
\begin{align}
	\overline{\mathcal{C}}_{\tu{i}} & =
	\{[A\ B]=\zeta^\top: (\zeta-\overline{\zeta}_{\tu{i}})^\top \overline{P}_{\tu{i}}^{-2} (\zeta-\overline{\zeta}_{\tu{i}}) \preceq \overline{Q}_{\tu{i}} \}. \label{setOverlineI-2}
\end{align}
In summary, because the form \eqref{setOverlineI-2} allows us to obtain our main result, we over-approximate $\mathcal{C}_{\tu{i}}$ in~\eqref{inter Ci set} through $\overline{\mathcal{C}}_{\tu{i}}$ in~\eqref{setOverlineI} where the matrices $A_{\tu{i}} \succ 0$ and $B_{\tu{i}}$ are determined by solving an optimization program, which we recall from~\cite[\S~5.1]{bisoffi2021tradeoffs}.

From data, define for $j = 0, \dots, T-1$
\begin{align} 
	& c_j :=-\omega I + \dot{x}^j  (\dot{x}^{j})^\top, \notag \\
	& b_j := -\bmat{z^j\\ v^j} (\dot{x}^{j})^\top,a_j := \bmat{z^j\\v^j}\bmat{z^j\\v^j}^\top.	\label{params of set inst bound single point}
\end{align}
As it is done for classical ellipsoids \cite[\S3.7.2]{boyd1994linear}, we impose that the matrix ellipsoid $\overline{\mathcal{C}}_{\tu{i}}$, which is well-defined for $A_{\tu{i}} \succ 0$, includes $\mathcal{C}_{\tu{i}}$ through the (lossy) S-procedure \cite[\S 2.6.3]{boyd1994linear} and we then minimize the size of $\overline{\mathcal{C}}_{\tu{i}}$. This corresponds to the optimization program
\begin{equation}\label{overapprox set I}
	\begin{split}
		& \text{minimize}\  -\log\det A_{\tu{i}} \\
		&\text{subject to} \\
		&\  \bmat{
			\vspace*{1pt}
			- I - \sum_{j =0}^{T-1} \tau_j c_j   &
			B_{\tu{i}}^\top - \sum_{j =0}^{T-1} \tau_j b_j^\top &
			B_{\tu{i}}^\top\\
			\vspace*{1pt}
			B_{\tu{i}} - \sum_{j =0}^{T-1} \tau_j b_j &
			A_{\tu{i}} - \sum_{j =0}^{T-1} \tau_j a_j &
			0\\
			B_{\tu{i}} &
			0 &
			- A_{\tu{i}}
		} \preceq 0, \\
		&\ A_{\tu{i}} \succ 0,\  \tau_j\ge 0 \text{ for } j = 0, \dots, T-1.
	\end{split}
\end{equation}
When this optimization program is solved, we use the returned $A_{\tu{i}}$ and $B_{\tu{i}}$ to obtain the matrices $\overline{\zeta}_{\tu{i}}$, $\overline{P}_{\tu{i}}$ $\overline{Q}_{\tu{i}}$ as in~\eqref{ell matr equiv}.
Before further analyzing the optimization program, we discuss in the next remark an alternative bound on the disturbance that is commonly used.

\begin{remark}
	\label{remark:energy bound}
	When an instantaneous bound on the disturbance is given by $|d|^2 \le \omega$ as in $\mathcal{D}_{\tu{i}}$ in~\eqref{in:bound}, one can infer that the whole unknown disturbance sequence $\bmat{d^0 & \dots  & d^{T-1}}$ of the experiment belongs to the set
	\begin{equation}
		\label{en:bound}
		\mathcal{D}_{\tu{e}} := \{D \in \mathbb{R}^{n\times T} : DD^{\top}  \preceq T \omega I \},
	\end{equation}
	which we call an \emph{energy} bound on the disturbance. By collecting data points in~\eqref{data instant} in matrices
	\begin{subequations}
		\label{data}
		\begin{align}
			X_1 & := \bmat{\dot{x}^0 & \dots  & \dot{x}^{T-1}}, \label{data:matX1}\\
			Z_0 & := \bmat{z^0 & \dots  & z^{T-1} } , V_0 := \bmat{v^0 & \dots & v^{T-1}},\label{data:matU}
		\end{align}
		the set of matrices consistent with data is
	\end{subequations}	
	\begin{equation}
		\label{Ce definition}
		\begin{split}
			\! & \mathcal{C}_{\tu{e}} := \Big\{ [A\ B] : X_1 = AZ_0 + B V_0 + D, \\ 
			& \hspace*{25mm} D \in \mathbb{R}^{n\times T}, DD^{\top}  \preceq T \omega I \Big\},
		\end{split}	
	\end{equation}		
	or, after standard manipulations as in \cite[Section 2.3]{bisoffi2021petersen},
	\begin{align}
		& \mathcal{C}_{\tu{e}} = \Big\{ [A\ B] = \zeta^\top \colon 
		\bmat{I & \zeta^\top}
		\bmat{C_{\tu{e}} & B_{\tu{e}}^{\top}\\B_{\tu{e}} & A_{\tu{e}}}
		\bmat{I \\ \zeta} \preceq 0 \} \notag \\
		& A_{\tu{e}} := \bmat{Z_0 \\ V_0} \bmat{Z_0 \\ V_0}^\top,\
		B_{\tu{e}}  := -\bmat{Z_0 \\ V_0} X_1^{\top}, \notag \\
		& C_{\tu{e}}  := - T \omega I + X_1X_1^{\top}.     \label{Ae Be Ce}
	\end{align}
	This form is mathematically analogous to that of $\overline{\mathcal{C}}_{\tu{i}}$ in~\eqref{setOverlineI} and one can indeed adapt the results we will find for $\overline{\mathcal{C}}_{\tu{i}}$ to the set $\mathcal{C}_{\tu{e}}$, as we will show in Remark~\ref{remark:energy bound cntd}.
	However, \cite{bisoffi2021tradeoffs} illustrates that, unless $T$ is very large and thus impacts the computational cost, it is advantageous to work with $\overline{\mathcal{C}}_{\tu{i}}$ instead of $\mathcal{C}_{\tu{e}}$.
\end{remark}

The results that follow rely on the optimization program in~\eqref{overapprox set I} being feasible.
We would like to show that feasibility of~\eqref{overapprox set I} is guaranteed under a relatively mild assumption.
This assumption is that the data set $\{\dot{x}^j, z^j, v^j\}_{j=0}^{T-1}$ yields a matrix $\smat{Z_0 \\ V_0}:=\smat{z^0 & \dots  & z^{T-1}\\v^0 & \dots & v^{T-1}}$ with full row rank.
This rank condition can be easily checked from data and when not verified, one can always collect additional data points, thereby adding columns to $\smat{Z_0 \\ V_0}$, to try and meet the rank condition (in a linear setting, this rank condition is related to classical persistence of excitation, see \cite[Section~4.1]{bisoffi2021petersen}).
We have then the next result for feasibility of \eqref{overapprox set I}.

\begin{lemma}
	\label{lemma:pers of exc implies opt feas}
	If $\smat{Z_0 \\ V_0}$ has full row rank, then the optimization program \eqref{overapprox set I} is feasible.
\end{lemma}
\begin{proof}
With the matrices $A_{\tu{e}}$ and $B_{\tu{e}}$ defined in~\eqref{Ae Be Ce} from data, set in~\eqref{overapprox set I}
	\begin{equation}
		\label{selection for feasib}
		\tau_0 = \dots = \tau_{T-1} = \tau, A_{\tu{i}} = \tau A_{\tu{e}}, B_{\tu{i}} = \tau B_{\tu{e}}
	\end{equation}
	for $\tau > 0$ to be determined in the proof.
	Since $\smat{Z_0 \\ V_0}$ has full row rank, $A_{\tu{e}}$ in~\eqref{Ae Be Ce} satisfies $A_{\tu{e}} \succ 0$.
	Consider the constraints in~\eqref{overapprox set I}; by \eqref{selection for feasib}, $A_{\tu{e}} \succ 0$ and the selection $\tau > 0$, the statement of the lemma is proven if we choose $\tau > 0$ suitably to verify
	\[
	\smat{
		- I -  \sum_{j =0}^{T-1} \tau c_j   &
		\tau B_{\tu{e}}^\top - \sum_{j =0}^{T-1} \tau b_j^\top & \tau B_{\tu{e}}^\top\\
		\tau B_{\tu{e}} - \sum_{j =0}^{T-1} \tau b_j & \tau A_{\tu{e}} - \sum_{j =0}^{T-1}  \tau a_j & 0\\
		\tau B_{\tu{e}} & 0 & - \tau A_{\tu{e}}
	} \preceq 0.
	\]
	By recalling \eqref{params of set inst bound single point} and \eqref{Ae Be Ce}, we have the relations $\sum_{j =0}^{T-1} c_j = -T \omega I + X_1 X_1^\top = C_{\tu{e}}$, $\sum_{j =0}^{T-1}  b_j= - \smat{Z_0 \\ V_0} X_1^\top =  B_{\tu{e}}$, $\sum_{j =0}^{T-1} a_j = \smat{Z_0 \\ V_0} \smat{Z_0 \\ V_0}^\top =  A_{\tu{e}}$. 
	By substituting these relations in the previous matrix inequality, we want to choose $\tau > 0$ suitably to verify
	\begin{equation*}
		\smat{
			- I - \tau  C_{\tu{e}}  & 0 & \tau B_{\tu{e}}^\top\\
			0 & 0 & 0\\
			\tau B_{\tu{e}} & 0 & - \tau A_{\tu{e}}
		} \preceq 0 \iff
		\smat{
			- I - \tau  C_{\tu{e}} + \tau B_{\tu{e}}^\top A_{\tu{e}}^{-1} B_{\tu{e}} & 0 \\
			0 & 0 } \preceq 0
	\end{equation*}
	by Schur complement and $\tau A_{\tu{e}} \succ 0$.
	The last condition is indeed true and the statement is thus proven because $\overline{Q}_{\tu{e}} := B_{\tu{e}}^\top A_{\tu{e}}^{-1} B_{\tu{e}} - C_{\tu{e}} \succeq 0$ by \cite[Lemma~1]{bisoffi2021petersen} and a sufficiently small $\tau > 0$ ensures $- I + \tau  \overline{Q}_{\tu{e}} \preceq 0$.
\end{proof}

With the set $\mathcal{C}_{\tu{i}}$ of matrices consistent with data in~\eqref{inter Ci set} and its over-approximation $\overline{\mathcal{C}}_{\tu{i}}$ obtained via the optimization program~\eqref{overapprox set I}, we can solve the considered problem of enforcing invariance for ground truth matrices $(\sa{A},\sa{B})$.
This is achieved by enforcing invariance in~\eqref{invCond} for all matrices $(A,B)$ in $\overline{\mathcal{C}}_{\tu{i}}$ as in the next main result, so that an invariant set is determined directly from data.

\begin{theorem}[Data-driven invariance condition]\label{th:main}
	For a design parameter $\epsilon > 0$, consider the data generation mechanism \eqref{data gen mech}, measured data $\{ \dot{x}^j, z^j, v^j\}_{j=0}^{T-1}$ as in~\eqref{data instant} and disturbance $d$ satisfying the instantaneous bound $\mathcal{D}_{\tu{i}}$ in~\eqref{in:bound} (i.e., $d^0 \in \mathcal{D}_{\tu{i}}$, \dots, $d^{T-1} \in \mathcal{D}_{\tu{i}}$). \newline
	Assume that the optimization program in~\eqref{overapprox set I} is feasible so that parameters $\overline{\zeta}_{\tu{i}}$, $\overline{P}_{\tu{i}}$ and $\overline{Q}_{\tu{i}}$ in~\eqref{ell matr equiv} exist.
	Assume that there exist decision variables $\ell \in \Pi $, $\eta \in \Pi $, $h \in \Pi$ and $K \in \Pi_{m,1}$ such that for all $x \in \mathbb{R}^n$, $\eta(x) > 0$ and $H(x) \preceq 0$, with $H$ defined as
\begingroup%
\thinmuskip=.5mu plus 1mu minus 1mu%
\medmuskip=1.0mu plus 1mu minus 1mu%
\thickmuskip=1.5mu plus 1mu minus 1mu%
\setlength\arraycolsep{1.5pt}%	
\begin{align}
		& H(x):=\!\bmat{\left\{ \begin{matrix}
				\ell(x)h(x) + \epsilon \hfill \\
				\!+ \hJac \overline{\zeta}_{\tu{i}}^\top\! \smat{Z(x)\\ W(x) K(x)}\!
			\end{matrix} \right\}  & \star & \star \\
			\eta(x) \overline{P}_{\tu{i}} \smat{Z(x)\\ W(x) K(x)} & -2 \eta(x) I & \star \\
			\overline{Q}_{\tu{i}}^{1/2} \hJac^\top & 0 & -2 \eta(x) I}.
		\label{qlinRobInvCond3}
\end{align}%
\endgroup%	
Then, the set $\mathcal{I}$ in~\eqref{eq:invset} is invariant for the system in~\eqref{polynSysClosed}.
\end{theorem}%

\begin{proof}
	Under the assumption that \eqref{overapprox set I} is feasible, $A_{\tu{i}} \succ 0$ by construction, so $\overline{P}_{\tu{i}}$ and $\overline{Q}_{\tu{i}}$ in~\eqref{ell matr equiv} satisfy $\overline{P}_{\tu{i}}\succ 0$ and $\overline{Q}_{\tu{i}}\succ 0$.
	These two relations allow us to rewrite \eqref{setOverlineI-2} as
	\begin{align}
		\hspace*{-2mm}\overline{\mathcal{C}}_{\tu{i}} &  \!= \{ \zeta^\top \!\! \in \real^{n \times (n+m)} \!
		\colon \! \overline{Q}_{\tu{i}}^{-1/2} (\zeta-\overline{\zeta}_{\tu{i}} )^\top \overline{P}_{\tu{i}}^{-1} \cdoT [\star]^\top \! \preceq \! I \}\! \label{setOverlineI-3} \\
		& \! = \{ (\overline{\zeta}_{\tu{i}} + \overline{P}_{\tu{i}} Y \overline{Q_{\tu{i}}}^{1/2})^\top\! \colon \! Y 
		\in \real^{(n+m) \times n}, Y^\top Y \preceq \! I \} \label{ell matr P Zc Q matrix norm}
	\end{align}
	where \eqref{setOverlineI-3} is obtained from~\eqref{setOverlineI-2} since $\overline{Q}_{\tu{i}}^{-1/2}\succ 0$ and \eqref{ell matr P Zc Q matrix norm} is obtained by setting $Y = \overline{P}_{\tu{i}}^{-1}(\zeta-\overline{\zeta}_{\tu{i}})\overline{Q}_{\tu{i}}^{-1/2}$ in~\eqref{setOverlineI-3}.
	Since the disturbance $d$ satisfies the instantaneous bound $\mathcal{D}_{\tu{i}}$ in~\eqref{in:bound}, $(\sa{A},\sa{B}) \in \bigcap_{j=0}^{T-1} \mathcal{C}^j_{\tu{i}} = \mathcal{C}_{\tu{i}}$ and, thus, $(\sa{A},\sa{B}) \in \overline{\mathcal{C}}_{\tu{i}}$ in \eqref{ell matr P Zc Q matrix norm} since \eqref{overapprox set I} enforces $\mathcal{C}_{\tu{i}} \subseteq \overline{\mathcal{C}}_{\tu{i}}$ by construction.
	The invariance condition in~\eqref{invCond} imposed for all matrices in $\overline{\mathcal{C}}_{\tu{i}}$ reads
	\begin{align}
		\label{inv cond for all data consist matr}
		\!\!\!\!\forall [A\ B] \! \in \overline{\mathcal{C}}_{\tu{i}},\forall x \in \real^n,
		\left\{
		\begin{matrix}
			\ell(x) h(x) + \frac{\partial h}{\partial x}(x)\dot{x} + \epsilon \! \leq 0 \hfill \\
			\dot{x}= \! A Z(x) + B W(x) K(x). 
		\end{matrix}
		\right.\!
	\end{align}
	Since $[\sa{A}\ \sa{B}] \in \overline{\mathcal{C}}_{\tu{i}}$, if condition \eqref{inv cond for all data consist matr} holds, then the set $\mathcal{I}$ in~\eqref{eq:invset} is invariant for the ground truth system in~\eqref{polynSysClosed}. 
	Therefore, the proof is complete if we verify \eqref{inv cond for all data consist matr}, i.e., if we verify that
	\begin{align*}
		& \forall [A\ B]  \in \overline{\mathcal{C}}_{\tu{i}},\forall x \in \real^n, \notag \\
		& \ell(x) h(x) + \epsilon + \hJacS \smat{A & B} \smat{Z(x) \\ W(x) K(x)} \\
		&= \ell(x) h(x) + \epsilon  + \tfrac{1}{\sqrt 2} \smat{Z(x) \\ W(x) K(x)}^\top
		\smat{A^\top\\ B^\top} \tfrac{1}{\sqrt 2}  \hJacS^\top \notag \\
		& \hspace*{5mm} + \tfrac{1}{\sqrt 2} \hJacS \smat{A & B} \tfrac{1}{\sqrt 2} \smat{Z(x) \\ W(x) K(x)} \le 0.
	\end{align*}
	Rewrite this invariance condition in the compact form
	\begin{equation}
		\label{genConditionBeforePetersen}
		\begin{aligned}
			& \forall \zeta^\top \in \overline{\mathcal{C}}_{\tu{i}}, \forall x\in \real^n,\\
			&  W(x) + S (x) \zeta R(x) + R(x)^\top \zeta^\top S(x)^\top \le 0 
		\end{aligned}
	\end{equation}	
	after defining
	\begin{align*}
		& W(x) := \ell(x) h(x) + \epsilon \notag \\
		& S (x) :=  \tfrac{1}{\sqrt 2} \smat{Z(x) \\ W(x) K(x)}^\top,\ R(x) := \tfrac{1}{\sqrt 2} \hJacS^\top.
	\end{align*}
	\eqref{genConditionBeforePetersen}, and thus \eqref{inv cond for all data consist matr}, is equivalent, by \eqref{ell matr P Zc Q matrix norm}, to
	\begin{align}
			&\hspace*{-4mm} \forall Y \text{ with } Y^\top Y \preceq I, \, \forall x\in \real^n, \notag\\
			&\hspace*{-4mm} W(x) + S (x) (\overline{\zeta}_{\tu{i}} + \overline{P}_{\tu{i}} Y \overline{Q}_{\tu{i}}^{1/2}) R(x) \notag\\
			&\hspace*{-4mm} \hspace*{1mm} + R(x)^\top (\overline{\zeta}_{\tu{i}} + \overline{P}_{\tu{i}} Y \overline{Q}_{\tu{i}}^{1/2})^\top S(x)^\top \notag\\
			&\hspace*{-4mm} = W(x) +  S (x) \overline{\zeta}_{\tu{i}} R(x) \!+ \! R(x)^\top \overline{\zeta}_{\tu{i}}^\top S(x)^\top \notag\\
			&\hspace*{-4mm} \hspace*{1mm} + S (x) \overline{P}_{\tu{i}} Y \overline{Q}_{\tu{i}}^{1/2} R(x) \!+\! R(x)^\top \overline{Q}_{\tu{i}}^{1/2} Y^\top \overline{P}_{\tu{i}} S(x)^\top\!\! \le\! 0.\!\!\!\!\!\!\label{genConditionBeforePetersen2}
	\end{align}
	If there exist $\eta \in \Pi$ with $\eta(x) >0$ for all $x \in \real^n$, we have
	\begin{align*}
		& \forall Y \text{ with } Y^\top Y \preceq I,  \forall x \in \real^n, \\
		&S (x) \overline{P}_{\tu{i}} Y \overline{Q}_{\tu{i}}^{1/2} R(x)   + R(x)^\top \overline{Q}_{\tu{i}}^{1/2} Y^\top \overline{P}_{\tu{i}} S(x)^\top\\
		& \le  \eta(x) S(x) \overline{P}_{\tu{i}}\, \overline{P}_{\tu{i}} S(x)^\top+
		\tfrac{1}{\eta(x)} R(x)^\top \overline{Q}_{\tu{i}}^{1/2} Y^\top Y \overline{Q}_{\tu{i}}^{1/2} R(x)\\
		& \le  \eta(x) S(x) \overline{P}_{\tu{i}}^2 S(x)^\top+
		\tfrac{1}{\eta(x)} R(x)^\top \overline{Q}_{\tu{i}} R(x)
	\end{align*}
	where we use Young's inequality in the first upperbound and $Y^\top Y \preceq I$ in the second upperbound. 
	Using this last upperbound in~\eqref{genConditionBeforePetersen2}, we obtain that \eqref{genConditionBeforePetersen2} is implied by the existence of $\eta$, with $\eta(x)>0$ for all $x\in \real^n$, such that 
	\begin{align*}
		& \hspace*{0mm} \forall x\in \real^n, \, W(x) +  S (x) \overline{\zeta}_{\tu{i}} R(x) +  R(x)^\top \overline{\zeta}_{\tu{i}}^\top S(x)^\top  \\
		& \hspace*{12mm} + \eta(x) S (x) \overline{P}_{\tu{i}}^2 S (x)^\top + \tfrac{1}{\eta(x)}R(x)^\top \overline{Q}_{\tu{i}} R(x) \le 0 .  \notag 
	\end{align*} 
	Replacing the explicit expressions of $W(x), S(x), R(x)$ in this inequality, we conclude that the invariance condition \eqref{inv cond for all data consist matr} holds if for all $x \in \real^n$, $\eta(x)>0$ and
	\begingroup
	\thinmuskip=.6mu plus 1mu
	\medmuskip=1.2mu plus 1mu
	\thickmuskip=1.8mu plus 1mu
	\begin{align*}
		& 0 \ge \ell(x) h(x)  + \epsilon + \hJacS \overline{\zeta}_{\tu{i}}^\top \smat{Z(x)\\ W(x) K(x)} \notag \\
		& +  \tfrac{\eta(x)}{2} \smat{Z(x)\\ W(x) K(x)}^\top \overline{P}_{\tu{i}}^2 \smat{Z(x)\\ W(x) K(x)} + \tfrac{1}{2 \eta(x)} \hJacS \overline{Q}_{\tu{i}} \hJacS^\top. 
	\end{align*}
	\endgroup
	This last condition is equivalent to having for all $x \in \real^n$, $\eta(x) > 0$ and $ H(x) \preceq 0$, as one can easily verify applying Schur complement \cite[p.~28]{boyd1994linear} on $ H(x) \preceq 0$.
	Hence, \eqref{inv cond for all data consist matr} holds, as we needed to complete the proof.
\end{proof}

\begin{remark}
	\label{remark:energy bound cntd}
	In Remark~\ref{remark:energy bound}, we commented on the possibility of having an energy bound on the disturbance.
	In that case, the conclusion of Theorem~\ref{th:main} continues to hold if the hypothesis of Theorem~\ref{th:main} is slightly adapted as follows:
	\begin{quote}
		For a design parameters $\epsilon > 0$, consider the data generation mechanism \eqref{data gen mech}, measured data $\{ \dot{x}^j, z^j, v^j\}_{j=0}^{T-1}$ as in~\eqref{data instant} and disturbance $d$ satisfying the \emph{energy bound $\mathcal{D}_{\tu{e}}$ in~\eqref{en:bound}} (i.e., $\bmat{d^0 & \dots & d^{T-1}}  \in \mathcal{D}_{\tu{e}}$).\newline		
		\emph{Assume that $\smat{Z_0\\ V_0}$ has full row rank and parameters $\overline{\zeta}_{\tu{i}}$, $\overline{P}_{\tu{i}}$ and $\overline{Q}_{\tu{i}}$ are equal}, instead of~\eqref{ell matr equiv}, \emph{to respectively}
		\begin{equation}\label{en:conver redux}
			\overline{\zeta}_{\tu{e}}:=-A_{\tu{e}} ^{-1}B_{\tu{e}} ,
			\overline{P}_{\tu{e}}:=A_{\tu{e}}^{-1/2},
			\overline{Q}_{\tu{e}}:=B_{\tu{e}} ^{\top}A_{\tu{e}} ^{-1} B_{\tu{e}} - C_{\tu{e}}
		\end{equation}	
		for $A_{\tu{e}}$, $B_{\tu{e}}$, $C_{\tu{e}}$ in~\eqref{Ae Be Ce}. Assume that there exist decision variables $\ell \in \Pi $, $\eta \in \Pi $, $h \in \Pi$ and $K \in \Pi_{m,1}$ such that for all $x\in \mathbb{R}^n$, $\eta(x) > 0$ and $H(x) \preceq 0$, with $H$ defined in~\eqref{qlinRobInvCond3}, \emph{where $\overline{\zeta}_{\tu{i}}$, $\overline{P}_{\tu{i}}$, $\overline{Q}_{\tu{i}}$ are respectively equal to~\eqref{en:conver redux}}.\newline
		Then, the set $\mathcal{I}$ in~\eqref{eq:invset} is invariant for the system in~\eqref{polynSysClosed}.
	\end{quote}
	This adaptation of Theorem~\ref{th:main} is proven by observing that: (i)~$\overline{P}_{\tu{e}}$ and  $\overline{Q}_{\tu{e}}$, which are used in place of respectively $\overline{P}_{\tu{i}}$ and $\overline{Q}_{\tu{i}}$, satisfy $\overline{P}_{\tu{e}} \succ 0$ and  $\overline{Q}_{\tu{e}} \succeq 0$ by the full row rank of $\smat{Z_0\\ V_0}$ and \cite[Lemma~1]{bisoffi2021petersen}; (ii)~the set $\mathcal{C}_{\tu{e}}$ in~\eqref{Ce definition}, resulting from $\mathcal{D}_{\tu{e}}$, can be written in the form~\eqref{ell matr P Zc Q matrix norm} by \cite[Prop.~1]{bisoffi2021petersen}; (iii)~the rest of the proof of Theorem~\ref{th:main} is the same.
\end{remark}

As discussed in Section~\ref{sec:intro}, the invariance condition obtained in Theorem~\ref{th:main} is instrumental to design safe control laws in applications.
To effectively link invariance and safe control we introduce the so-called safe set $\mathcal{S}$, by which a user can specify all constraints on the state.
Formally, these constraints are expressed by positivity of polynomials $\sigma_1$, \dots, $\sigma_	q$ ($q \in \mathbb{Z}_{\ge 1}$) and the safe set $\mathcal{S}$ is
\begin{equation}
	\label{safe set}
	\mathcal{S}:= \{ x \in \real^n \colon \sigma_j(x) \le 0, j =1,\dots, q\}.
\end{equation}
Hence, when designing the invariance set $\mathcal{I}$, one needs to enforce the condition $\mathcal{I} \subseteq \mathcal{S}$ so that when the state belongs to $\mathcal{I}$, it will also comply with all constraints expressed by $\mathcal{S}$.
At the same time, it is of interest not only to impose $\mathcal{I} \subseteq \mathcal{S}$, but to ensure that $\mathcal{I}$ is as ``large''  as possible.
Using a classical approach as in, e.g., \cite{jarvis2005control}, define the set $\mathcal{L}_\theta$ for a nonnegative polynomial $\lambda$ (i.e., $\lambda(x) \ge 0$ for all $x$) and a nonnegative scalar $\theta$ as
\begin{equation}
	\label{variable-size set}
	\mathcal{L}_\theta : = \{ x \in \real^n \colon \lambda(x) \le \theta \}.
\end{equation}
With $\mathcal{L}_\theta$, $\mathcal{I}$ can be enlarged by imposing $\mathcal{L}_\theta \subseteq \mathcal{I}$ while maximizing $\theta \ge 0$; hence, $\mathcal{L}_\theta$ acts as a variable-size set that dilates $\mathcal{I}$ from the inside according to the shape given by the design parameter $\lambda$, which can be chosen based on the form of the safe set $\mathcal{S}$. This approach will be exemplified in Section~\ref{sec:platoon}.

Moreover, we strengthen the positivity conditions of Theorem~\ref{th:main} into more conservative, but computationally tractable, sum-of-squares conditions.
Through the safe set $\mathcal{S}$, the variable-size set $\mathcal{L}_\theta$, the strengthening of positivity conditions of Theorem~\ref{th:main} and definition
\[
r:=1+m+2n,
\]
we have the next result for data-driven safe control.
\begin{theorem}[Data-driven safe control]
	\label{cor:extra conditions}
	For given polynomials $\sigma_1$, \dots, $\sigma_q$ and nonnegative polynomial $\lambda$, consider the set $\mathcal{S}$ in~\eqref{safe set} and the set $\mathcal{L}_\theta$ in~\eqref{variable-size set} parametrized by $\theta$, along with $H$ defined in~\eqref{qlinRobInvCond3} and a given $\bar \eta > 0$.\newline
	Under the same hypothesis of Theorem~\ref{th:main}, assume the program
	\begin{subequations}\label{cor cond}%
		\begin{align}%
			& \hspace*{-1mm} \text{maximize} && \theta\\
			& \hspace*{-1mm} \text{s.t.} && \hspace*{-11mm} \{ \ell , \eta , h \} \!\subseteq \Pi, K\! \in \Pi_{m,1}, \{ s_1, \dots, s_q, \varsigma \}\! \subseteq \Sigma, \theta\! \ge 0\label{cor cond:set defining variables}\\
			&\hspace*{-1mm} && \hspace*{-11mm} \eta -\bar \eta \in \Sigma, - H \in \Sigma_r \label{cor cond:same as thm}\\
			&\hspace*{-1mm} && \hspace*{-11mm} \varsigma ( \lambda - \theta) - h \in \Sigma, \label{cor cond:var-sized set contained in inv set}\\
			&\hspace*{-1mm} && \hspace*{-11mm} s_j  h - \sigma_j \in \Sigma, j =1, \dots, q. \label{cor cond:inv set contained in safe set}
		\end{align}%
	\end{subequations}%
	 has a solution. Then, the set $\mathcal{I}$ in~\eqref{eq:invset} is invariant for the system in~\eqref{polynSysClosed} and satisfies $\mathcal{L}_\theta \subseteq \mathcal{I} \subseteq \mathcal{S}$. 
\end{theorem}

\begin{table*}[htb]
	\centering
	\begin{tabular}{c|c|c|c|c|c|c|c|c|c|c|c}
		$\gamma_1$ & $\beta_1$ &  $\alpha_1$ & $\gamma_2$ & $\beta_2$ & $\alpha_2$ & $\tau_{\tu{h}}$ & $d_0$ & $d_1$ & $v_{\tu{M}}$ & $\bar{v}$ & $\bar{d}$  \\ \hline
		0.005$\frac{\text{N}}{\text{kg}}$ & 0.1$\frac{\text{N}\,\text{s}}{\text{kg}\,\text{m}}$ & 0.02$\frac{\text{N}\, \text{s}^2}{\text{kg}\,\text{m}^2}$ & 0.005$\frac{\text{N}}{\text{kg}}$ & 0.2 $\frac{\text{N}\, \text{s}}{\text{kg}\, \text{m}}$ & 0.04$\frac{\text{N}\, \text{s}^2}{\text{kg}\, \text{m}^2}$ & 0.2$\text{s}$ & 5$\text{m}$ & 10$\text{m}$ & 22.2$\frac{\text{m}}{\text{s}}$ & 8.5$\frac{\text{m}}{\text{s}}$ & 8$\text{m}$
	\end{tabular}
	\smallskip
	\caption{Values of the parameters for the two cars platoon simulations in Section~\ref{sec:platoon}.}
	\label{tab:Fx}
\end{table*}

\begin{proof}
	\eqref{cor cond:same as thm} implies that for all $x\in \mathbb{R}^n$, $\eta(x) \ge \bar \eta > 0$ and $H(x) \preceq 0$ (by Definition~\ref{def:matrixSOS}) and, under the same hypothesis of Theorem~\ref{th:main}, these two conditions were shown in Theorem~\ref{th:main} to imply that $\mathcal{I}$ is invariant for~\eqref{polynSysClosed}.
	The statement is then proven if we show that with~\eqref{cor cond:set defining variables}, \eqref{cor cond:var-sized set contained in inv set} implies $\mathcal{L}_\theta \subseteq \mathcal{I}$ and \eqref{cor cond:inv set contained in safe set} implies $\mathcal{I} \subseteq \mathcal{S}$.
	Since the reasoning is the same, we show only the latter.
	$\mathcal{I} \subseteq \mathcal{S}$ is equivalent to the set inclusion
	\begin{equation*}
		\{ x \in \mathbb{R}^n \colon h(x) \leq 0 \} \subseteq \{ x \in \mathbb{R}^n \colon \sigma_j(x) \leq 0 \}
	\end{equation*}
	holding for all $j = 1,\dots,q$ or, equivalently, to the empty-set condition
	\begin{equation*}
		\{ x \in \mathbb{R}^n \colon -h(x) \geq 0, \sigma_j(x) \geq 0, \sigma_j(x) \neq 0 \} = \emptyset
	\end{equation*}
	holding for all $j = 1,\dots,q$
	or, equivalently by Fact~\ref{fact:psatz}, to the existence, for all $j = 1,\dots,q$, of polynomials $s_{j,0}$, $s_{j,1}$, $s_{j,2}$, $s_{j,3}$ in $\Sigma$ and $k_j \in \mathbb{Z}_{\ge 0}$, $j = 1,\dots,q$ such that 
	\begin{equation*}
		s_{j,0} - s_{j,1} h + s_{j,2}\sigma_j - s_{j,3} h \sigma_j+\sigma_j^{2k_{j}} = 0.
	\end{equation*}
	This is implied by the existence, for $j = 1$,\dots, $q$, of $k_j=1$, $s_{j,0}=0$, $s_{j,1}=0$ and $s_{j,2}$, $s_{j,3}$ in $\Sigma$ such that $\sigma_j(s_{j,2}-s_{j,3} h + \sigma_j) = 0$, which is implied by~\eqref{cor cond:inv set contained in safe set}.
\end{proof}

	To conclude the section, we show in the next remark how input constraints can be readily incorporated in the proposed design to account for actuator limitations.
	
	\begin{remark}
		Suppose that input $u$ needs to be bounded in norm, i.e., $|u| \le u_{\tu{M}}$ for some $u_{\tu{M}} >0$.
		This constraint is enforced by imposing that for each $x$, $|K(x)| \le u_{\tu{M}}$, which is equivalent to $K(x)^\top K(x) \le u_{\tu{M}}$ and $\smat{u_{\tu{M}} & K(x)^\top \\ K(x) & I} \succeq 0$.
		This condition can be relaxed as $\smat{u_{\tu{M}} & K^\top \\ K & I} \in \Sigma_{m+1}$ and added to the conditions in Theorem~\ref{cor:extra conditions}.
	\end{remark}%

\section{Numerical example: car platooning}
\label{sec:platoon}
As a safety-critical system, we consider two cars moving in a platoon formation.
The system can be modeled as
\begin{subequations}
	\label{platoon:physical coords}
	\begin{align}
		\dot x_1 & = u_1 -\gamma_1 -\beta_1 x_1 - \alpha_1 x_1^2\\
		\dot x_2 & = u_2 -\gamma_2 -\beta_2 x_2 - \alpha_2 x_2^2\\
		\dot x_3 & = x_1 - x_2
	\end{align}
\end{subequations}
where: the components $x_1$, $x_2$, $x_3$ of state $x$ represent respectively the velocity of the front vehicle, the velocity of the rear vehicle and the relative distance between the two; the components $u_1$, $u_2$ of input $u$ represent the forces normalized by vehicle mass; $\gamma_k$, $\beta_k$, $\alpha_k$ for $k=1,2$ are the static, rolling and aerodynamic-drag friction coefficients normalized by vehicle mass.
We impose safety constraints by the set
\begin{align}
	\mathcal{S} & := \{ x \in \mathbb{R}^3 \colon  d_0 + \tau_{\tu{h}} x_2 \leq x_3,\ x_3 \leq d_1, \notag \\
	& \hspace*{10mm} 0 \leq x_1,\ x_1 \leq v_{\tu{M}},\ 0 \leq x_2,\ x_2 \leq v_{\tu{M}} \} \notag \\
	&=: \{ x \in \mathbb{R}^3 \colon \sigma_1(x) \leq 0, \dots, \sigma_6(x) \leq 0 \}
	\label{eq:safeset}
\end{align}
where $d_0 + \tau_{\tu{h}} x_2$ is a relative distance required to avoid collisions with the front vehicle ($d_0 > 0$ is a standstill distance and $\tau_{\tu{h}} > 0$ is a time headway), $d_1$ is a distance required to keep the benefits of platooning (especially, aerodynamic drag reduction), and $v_{\tu{M}}$ is a maximum velocity allowed on the road.
Finally, we consider as point of interest $\bar x := (\bar v,\bar v,\bar d)$ where $\bar v$ is a predefined cruise velocity and $\bar d$ is a reference safety distance. 
Numerical values of parameters are in Table~\ref{tab:Fx}.

The numerical program to find invariant set and controller is in Algorithm~\ref{alg:cars} and is implemented in Matlab with YALMIP \cite{lofberg2004yalmip,Lofberg2009} and MOSEK. 
We now comment Algorithm~\ref{alg:cars}, which consists of an initialization (lines~1-2) and a main part (lines~3-15) made of two steps.

\begin{algorithm}[hbt!]
	\caption{Car Platooning}\label{alg:cars}
	\begin{algorithmic}[1]
		\State Excite the system around an equilibrium point of interest, collect input/state data and obtain a Lyapunov function for the linearized system from data by \cite[Th. 2]{bisoffi2021petersen}. 
		\State \textbf{initialize:} $\iota=0$ (a counter), $\theta_\iota=\theta_0=0.01$, $\epsilon=0.01$, $\eta(x) = 1$, $h$ equal to the Lyapunov function found above. 
		\Repeat
		\State Find $\ell \in \Pi$, $K \in \Pi_{m,1}$ and $s_{1}, \ldots, s_{6}, \varsigma \in \Sigma$
		\State subject to \quad $-H \in \Sigma_r$
		\State \phantom{subject to} \quad $\varsigma (\lambda -\theta_\iota) - h \in \Sigma$
		\State \phantom{subject to} \quad $s_j h -\sigma_j \in \Sigma, \quad j = 1,\dots,6$.
		\State Update $\ell$, $K$, $s_{1}, \ldots, s_{6}$, $\varsigma$. \\
		\State Maximize  $\theta$ \quad over $\eta, h\in \Pi$ and $\theta \ge 0$, 
		\State subject to \quad $\eta \in \Sigma,\ -H \in \Sigma_r$
		\State \phantom{subject to} \quad $\varsigma(\lambda -\theta)  - h  \in \Sigma$
		\State \phantom{subject to} \quad $s_j h -\sigma_j \in \Sigma, \quad j = 1,\dots,6$.
		\State Update $\eta$, $h$, $\theta_\iota = \theta$, $\iota \gets \iota+1$.
		\Until{$\theta_\iota - \theta_{\iota-1}$  is less than some tolerance ($10^{-3}$).}
	\end{algorithmic}
\end{algorithm}

As for the initialization of Algorithm~\ref{alg:cars}, we set the degrees of polynomials $h$, $\eta$, $s_j$ ($j = 1,\dots,6$) and $\varsigma$ to $4,2,2,2$.
Moreover, the shape of the safe set $\mathcal{S}$ in~\eqref{eq:safeset} (light blue set in Figure~\ref{fig:car2-c}) is wider in the coordinates $x_1$ and $x_2$ and narrow in the coordinate $x_3$.
Hence, we dilate the invariant set $\mathcal{I}$ through an ellipsoid shaped similarly to $\mathcal{S}$ and with center $\bar x$ since we would like $\mathcal{I}$ to contain $\bar x$; i.e., we dilate $\mathcal{I}$ through
\begin{equation*}
	\mathcal{L}_\theta : = \! \{ x \in \real^3 \colon ( x - \bar x)^\top\! \smat{0.02& 0 & 0\\ 0 & 0.05 & 0\\ 0 & 0 & 1} ( x - \bar x) =: \lambda(x) \le \theta \}
\end{equation*}
as in~\eqref{variable-size set}.
Finally, a guess of $h$ is needed to initialize the iterations in the procedure.
In a data-based fashion, we can use as a guess of $h$ the Lyapunov function obtained \emph{from data} by using \cite[Th.\ 2]{bisoffi2021petersen} and performing a preliminary experiment with ``small'' input and state signals around the equilibrium $\bar x$. 

As for the main part of Algorithm~\ref{alg:cars}, it corresponds to the SOS program \eqref{cor cond} in Theorem~\ref{cor:extra conditions}.
However, since \eqref{cor cond} presents products between decision variables, we first fix $h$, $\eta$ and $\theta$ in~\eqref{cor cond} and solve for the other decision variables (lines~4-8), and then fix $\ell$, $K$ and $s_1$, \dots, $s_6$, $\varsigma$ in~\eqref{cor cond} and solve for the other decision variables (lines~10-14).
Moreover, we asked $\eta - \bar \eta \in \Sigma$ in Theorem~\ref{cor:extra conditions} for (small) $\bar \eta > 0$; here, we ask the weaker condition $\eta \in \Sigma$ because, if the constraint $\eta \in \Sigma$ is feasible, interior-point algorithms automatically find \cite[p.~41]{ahmadi2008nonmonotonic} a strictly positive $\eta$, hence satisfying $\eta - \bar \eta \in \Sigma$, which we verified a-posteriori.
Finally, since $- H \in \Sigma_r$ in~\eqref{cor cond} is homogeneous with respect to $(h,\eta)$, we prune solutions by fixing the 0-degree coefficient of $h$ to a \emph{given} constant. 

We also remark that  in $H$ the terms $Z$ and $W$ appear. 
The choice of the monomials considered in~\eqref{polynSysOpen} for $Z$ and $W$ is important to obtain the best result from our solution. 
The simplest choice is to consider all monomials in $x_1$, $x_2$, $x_3$ up to a maximum degree; with noisy data, however, the coefficient of each monomial becomes uncertain and coping with it results in more conservative solutions. 
A smarter choice is to include high-level prior knowledge \cite{ahmadi2021learning}. 
For platooning, we use $Z(x) = \smat{x_1& x_2 & x_3 & x_1^2 & x_2^2}^\top$ and $W(x) = I$, since we know from physical considerations that the time derivative of the relative distance $x_3$ depends only on the velocities $x_1$ and $x_2$, and $\dot x_1$, $\dot x_2$ depend only on $x_1$, $x_1^2$, $x_2$, $x_2^2$ and $u$.

\begin{figure*}[!htb]
	\normalsize
	{
		\setcounter{MaxMatrixCols}{20}
		\begin{equation}
			\begin{split}
				h(x) &= 90.30 + 202.66x_1^2 + 364.18x_2^2 - 
				458.07x_1x_2 + 28.19x_1x_3 - 283.29x_2x_3 + 
				158.34x_3^2 - 0.95x_1^3  \\ 
				&- 0.63x_1^2x_2 - 47.29x_1^2x_3 + 3.70x_1x_2^2 - 3.89x_2^3 - 
				83.13x_2^2x_3 + 108.33x_1x_2x_3 - 
				7.36x_1x_3^2 + 65.74x_2x_3^2\\ 
				& - 38.24x_3^3 + 0.10x_1^4 - 0.10x_1^3x_2 + 
				0.36x_1^2x_2^2 - 0.25x_1x_2^3 + 
				0.29x_2^4 - 0.21x_1^3x_3 - 
				0.36x_1^2x_2x_3 - 0.43x_1x_2^2x_3  \\ 
				& - 0.46x_2^3x_3 +3.47x_1^2x_3^2 - 
				5.93x_1x_2x_3^2 + 6.15x_2^2x_3^2 - 
				0.18x_1x_3^3 - 4.95x_2x_3^3 + 2.75x_3^4.
			\end{split}	\label{eq:resulth}
		\end{equation}
	}
	\hrulefill
	\vspace*{4pt}
\end{figure*}

By using Algorithm \ref{alg:cars}, we obtain an invariant set $\mathcal{I} = \{ x \in \real^3 \colon h(x) {\color{cyan}} \leq 0 \}$ with $h$ as in~\eqref{eq:resulth}, displayed over two columns, and a controller $u = K(x) = \smat{K_1(x)\\ K_2(x)}$ with
\begingroup
\thinmuskip=1.mu plus 1mu
\medmuskip=2.mu plus 2mu
\thickmuskip=3.mu plus 3mu
\begin{align*}	
	K_1(x) &= 14.39x_3 -10.45x_1 + 44.05x_2 - 
	0.85x_1^2 - 4.24x_2^2, \\
	K_2(x) &=71.28x_3 -2.47x_1 - 23.32x_2 + 0.10x_1^2 - 4.09x_2^2.
\end{align*}
\endgroup
We removed monomials coefficients smaller than $10^{-5}$ in $h$ and $K$.
We compared our data-driven solution in Algorithm~\ref{alg:cars} against a model-based solution that knows perfectly the model, thereby providing a baseline for what we can achieve with the data-based scheme.
The model-based implementation, which is the counterpart of~\eqref{cor cond}, corresponds to
\begin{align*}
	& \text{maximize} && \theta\\
	& \text{s.t.} && \hspace*{-11mm} \{ \ell , h \} \subseteq \Pi, K \in \Pi_{m,1}, \{ s_1, \dots, s_6, \varsigma \} \subseteq \Sigma, \theta \ge 0\\
	& && \hspace*{-11mm}-(\ell h + \epsilon + \tfrac{\partial h}{\partial x} \smat{\sa{A} & \sa{B} } \smat{Z\\ W K}) \in \Sigma
	\\
	& && \hspace*{-11mm} \varsigma ( \lambda - \theta) - h \in \Sigma,\,\, s_j  h - \sigma_j \in \Sigma, j =1, \dots, 6. 
\end{align*}
In Figure~\ref{fig:car2-c}, we can see that the model-based and the data-driven solutions are comparable as for the sizes of the resulting invariant sets for data points affected by a disturbance satisfying $|d|\le 10^{-3}$. In both cases safety constraints are not violated since both invariant sets are within $\mathcal{S}$.
In Figure~\ref{fig:car2-phase}, the invariant set $\mathcal{I}$ of the data-driven solution is plotted together with trajectories of vector field \eqref{platoon:physical coords} in closed loop with controller $u = K(x)$. Trajectories are initialized close to the boundary of $\mathcal{I}$ to show that once in the set $\mathcal{I}$, they never leave it, thereby confirming that $\mathcal{I}$ is invariant.

\begin{figure}[!htb]
	\centering
	\includegraphics[width=1\linewidth]{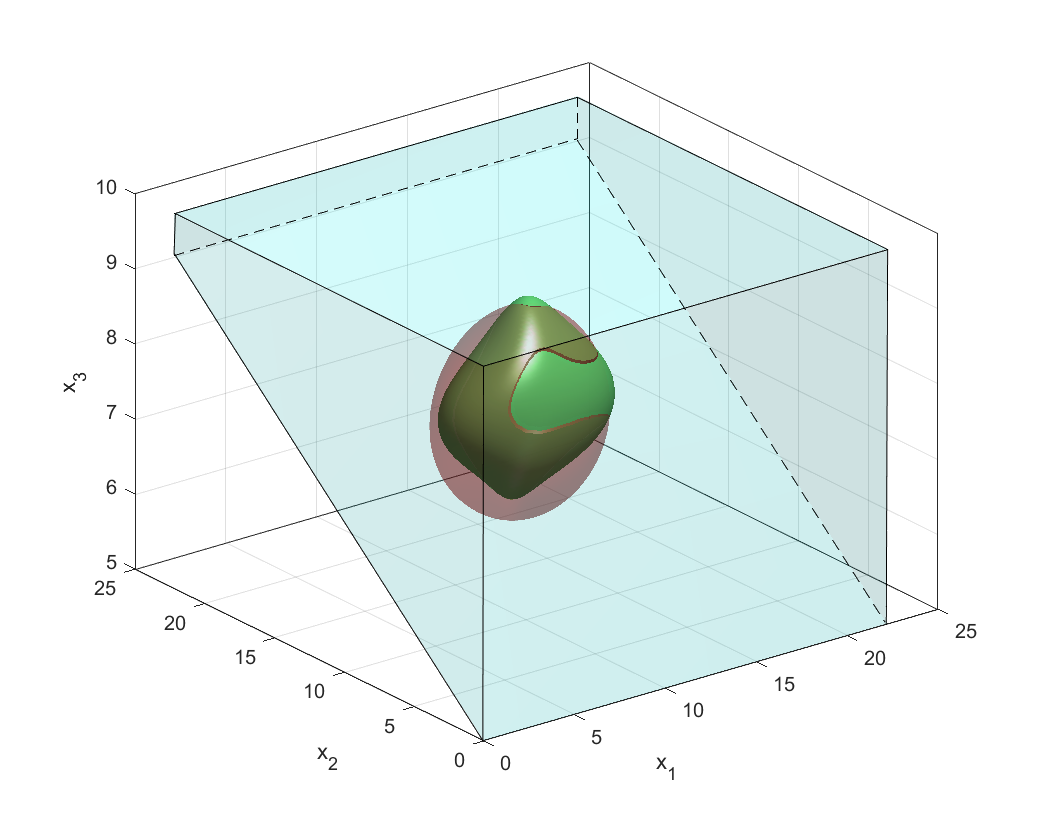}
	\caption{Invariant set for the model-based case in light red and for the data-driven case in green ($|d| \leq 10^{-3}$, $T=1000$). The safe set in light blue highlights that the two invariant sets comply with safety constraints.}
	\label{fig:car2-c}
\end{figure}

\begin{figure}
	\centering
	\includegraphics[width=1\linewidth]{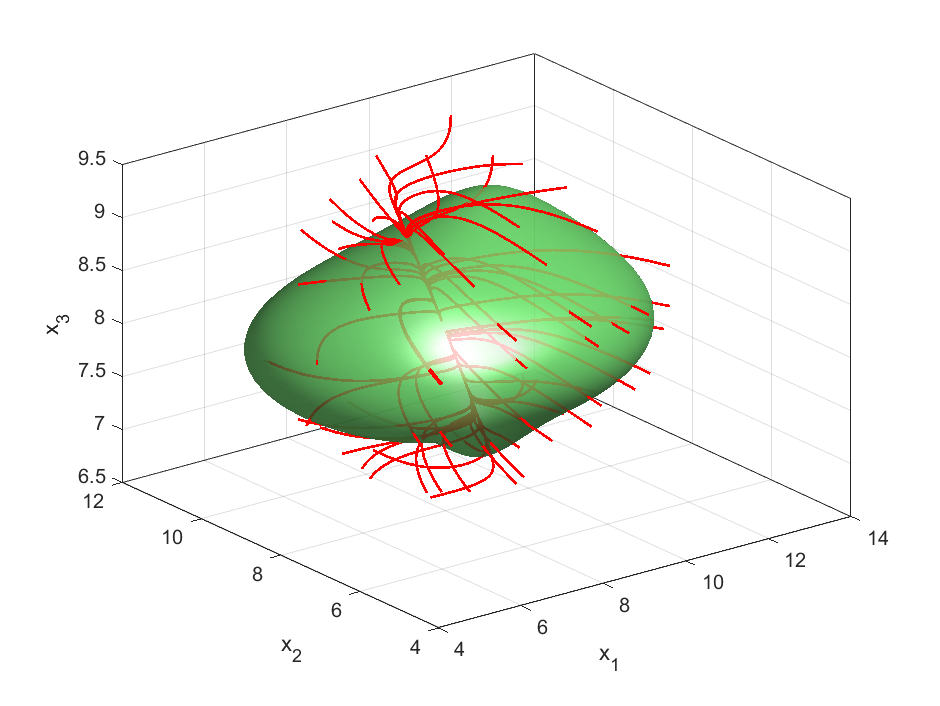}
	\caption{To show that the invariant set $\mathcal{I}$ (green) given by~\eqref{eq:resulth} is invariant, we simulate the closed loop system with initial conditions slightly outside $\mathcal{I}$, close to the boundary. 
	All simulated trajectories (red) never leave the set after entering it.}
	\label{fig:car2-phase}
\end{figure}

\section{Conclusions}
\label{sec:conclusion}

In this work we addressed the problem of enforcing invariance for a polynomial system based on data.
We assumed that data are corrupted by noise, whose nature and characteristic are unknown except for an (instantaneous) bound. 
The presence of noise resulted then in a set of dynamics consistent with data, which we over-approximated via matrix ellipsoids and took into account in the design.
Our solution provided a data-dependent SOS optimization program to obtain a state feedback controller and an invariant set for the closed loop system and we optimized the size of the invariant set  under the constraint that it remains contained in a user-defined safety set.
Finally, we tested our data-driven algorithm on a platooning example where we showed that, for a reasonable noise level, our solution compares well with the case of perfect model knowledge.
In this work we focused on enforcing invariance around one predefined state and it would be interesting to extend our data-driven design for safety around a prescribed trajectory in future research.

\bibliographystyle{plain}
\bibliography{references} 

\end{document}